\documentclass[12pt,a4paper]{article}
\newcommand{\argmax}{\operatornamewithlimits{arg\,max}}
\usepackage{mathrsfs}
\usepackage{amsmath}
\usepackage{dsfont}
\usepackage{amsfonts}
\usepackage{amssymb}
\usepackage{graphicx}
\usepackage{txfonts}
\newtheorem{proposition}{Proposition}

\newenvironment{proof}{\noindent\textit{Proof~~}}
{\nolinebreak[4]\hfill$\square$\\\par}
\title{On subgame perfect equilibria in quantum Stackelberg duopoly with incompete information}
\author{Piotr Fr\c{a}ckiewicz\\
\small Institute of Mathematics, Pomeranian University\\ \small 76-200 S\l upsk, Poland\\ \small fracor6@icloud.com}
\begin{document}
\maketitle
\begin{abstract}
The Li-Du-Massar quantum duopoly model is one of the generally accepted quantum game schemes. It has applications in a wide range of duopoly problems. Our purpose is to study Stackelberg's duopoly with incomplete information in the quantum domain. The result of Lo and Kiang has shown that the correlation of players' quantities caused by the quantum entanglement enhances the first-mover advantage in the game. Our work demonstrates that there is no first-mover advantage if the players' actions are maximally correlated. Furthermore, we proved that the second mover gains a higher equilibrium payoff that the first one.  
\end{abstract}
\section{Introduction}
Game theory, launched in 1928 by John von Neumann in \cite{neumann} and developed in 1944 by John von Neumann and Oskar Morgenstern in a book \cite{neumannmor} is one of the youngest branches of mathematics. The aim of this theory is to model mathematically the behavior of rational participants who aim at maximizing their own gain and take into account all possible ways of behaving of remaining participants.

The field developed on the border of game theory and quantum information theory is quantum game theory. This is an interdisciplinary area of research within which considered games are assumed to be played with the use of objects that behave according to the laws of quantum mechanics.

The first attempt to describe the game in the quantum domain applied to a simple coin tossing game \cite{meyer} and $2\times 2$ bimatrix games \cite{eisert}, \cite{marinatto}. Shortly after that quantum game theory has found applications in various fields including decision sciences \cite{khrennikov}, \cite{busemeyer}, \cite{fracorimperfectrecall}, financial theory \cite{piotrowski1}, \cite{piotrowski2}, \cite{piotrowski3} or mathematical psychology \cite{busemeyer}. 

A lot of attention has been focused  on the duopoly problems. One of the generally accepted quantum duopoly scheme is due to Li et al. \cite{lidumassar}. A rich literature applies the Li-Du-Massar scheme to the Cournot dupoly problems \cite{chen}, \cite{yli}, \cite{sekiguchi}, \cite{fracorremark}, the Stackelberg duopoly \cite{previousstackelberg}, \cite{wang}, \cite{ryu}, \cite{xia} and the Bertrand duopoly examples \cite{lobertrand}, \cite{qinbertrand}, \cite{fracorbertrand}

The existing results motivate further study rather than exhaust the subject. Our previous work \cite{fracorremark} shows that the quantum Cournot duopoly given by a piecewise function requires a best reply analysis to determine Nash equilibria of the game. This method found further application in the quantum Bertrand duopoly (with discontinuous payoff functions) \cite{fracorbertrand}. Another problem worth restudying is the Stackelberg duopoly. The Li-Du-Massar approach to this sequential type of duopoly was first investigated in \cite{previousstackelberg}. The paper provides a comprehensive analysis for the entanglement parameter $\gamma$ bounded by $\gamma_{0} = (1/2)\sinh^{-1}(1)$. We learned from \cite{previousstackelberg} that the first-mover advantage is enhanced when $\gamma >0$, and the difference between the equilibrium payoffs of player 1 and 2 grows monotonically with $\gamma$, having the maximum value at $\gamma = \gamma_{0}$. Surprisingly, we showed that the first-mover advantage is suppressed as the entanglement parameter $\gamma$ goes to infinity, while the second player's equilibrium payoff increases, and converges to the first player's equilibrium payoff \cite{fracorduopoly2}.

The aim of this paper is to examine the quantum Stackelberg duopoly with incomplete information. Our work is a reexamination of the problem first formulated in \cite{lokiang}. We show that the previous result only partially answers the question what optimal strategies and the corresponding payoff results are in both the classical and quantum game. The equilibrium outcome found in \cite{lokiang} turns out to be well defined only for initial values of $\gamma$. As a result, it is still an open question what the equilibrium result is, for example, in the most interesting case $\gamma \to \infty$.

We are interested in finding subgame perfect equilibria of the game and the related equilibrium outcomes. It needs a more sophisticated reasoning compared with \cite{lokiang}. Our method is based on applying best reply analysis depending on different choices of the players.
\section{Brief review of quantum Stackelberg duopoly with incomplete information}
The problem of quantum approach to Stackelberg's duopoly with incomplete information was first investigated in \cite{lokiang}. The authors used the Li-Du-Massar scheme \cite{lidumassar} to study Stackelberg's duopoly with an additional assumption that a player who moves first (player A) is uncertain about the marginal cost of the other player (player B). To be specific, player A moves first and offers quantity $x_{A}$ of a homogeneous product. Player B observes the move and then  offers the quantity $x_{B}$ of the product. In a game of incomplete information players may be uninformed about certain characteristics of the game. In the case studied in \cite{lokiang}, player A has incomplete information about player B's marginal cost $c_{B}$. Player A only knows that $c_{B} = c_{BH}$ occurs with probability $\theta$ and $c_{B} = c_{BL}$ occurs with probability $1-\theta$. In contrast, player B knows her own marginal cost and that of player A. It was assumed in \cite{lokiang} that player A's payoff function is 
\begin{equation}\label{payoffA}
u_{A}(x_{A},x_{B}) = x_{A}(k - x_{A} - x_{B}),
\end{equation}
and depending on player B's marginal cost, her payoff function is
\begin{equation}\label{payoffB}
u_{BL}(x_{A}, x_{B}) = x_{B}(k_{L}-x_{A} -x_{B}) \quad \text{or} \quad u_{BH}(x_{A}, x_{B}) = x_{B}(k_{H}-x_{A} -x_{B}),
\end{equation}
where $k_{L} = a - c_{BL}$, $k_{H} = a-c_{BH}$, $k = \theta k_{H} + (1-\theta) k_{L}$, $k_{H} > k_{L} > 0$, and $a$ is the price a consumer is willing to pay for the product if there are no products on the market. It was shown that the Nash equilibrium outcome $(x^*_{A}, x^*_{BH}, x^*_{BL})$ is given by
\begin{equation}\label{classicsolutionfake}
x^*_{A} = \frac{1}{2}k, \quad x^*_{BH} = \frac{1}{2}\left(k_{H} - \frac{k}{2}\right), \quad x^*_{BL} = \frac{1}{2}\left(k_{L} - \frac{k}{2}\right)
\end{equation}
with the payoff outcomes
\begin{equation}\label{classicsolutionfakeo}
u_{A}(x^*_{A}, x^*_{BH}, x^*_{BL}) = \frac{1}{8}k^2, \quad u_{BH}(x^*_{A}, x^*_{BH}) = \frac{1}{4}\left(k_{H}-\frac{k}{2}\right)^2, \quad u_{BL}(x^*_{A}, x^*_{BL}) = \frac{1}{4}\left(k_{L}-\frac{k}{2}\right)^2.
\end{equation}
Solution (\ref{classicsolutionfake}) is valid on condition that 
\begin{equation}\label{classicassumption}
k_{L} \geq \frac{k}{2} \Leftrightarrow \theta \leqslant \frac{k_{L}}{k_{H} - k_{L}}.
\end{equation}
The quantum extension of (\ref{payoffA}) and (\ref{payoffB}) was obtained by replacing $x_{A}$ and $x_{B}$ with $x_{A}\cosh\gamma + x_{B}\sinh\gamma$ and $x_{B}\cosh\gamma + x_{A}\sinh\gamma$, respectively. The resulting Nash equilibrium outcome was found to be
\begin{align}\label{qsolution1}
x^*_{A} &= \frac{\cosh^2\gamma \exp(-\gamma)}{1+ \cosh\gamma \exp(-\gamma)}k,\\ x^*_{BL} &= \frac{1}{2}\exp(-\gamma)\left(k_{L} - \frac{\cosh\gamma \exp \gamma}{1+ \cosh\gamma \exp(-\gamma)}k \right)  \label{qsolution2},\\ 
x^*_{BH} & = \frac{1}{2}\exp(-\gamma)\left(k_{H} - \frac{\cosh\gamma \exp \gamma}{1+ \cosh\gamma \exp(-\gamma)}k \right),\label{qsolution3}
\end{align}
providing that
\begin{equation}\label{quantumrestriction}
k_{L} \geq \frac{1+ \exp(-2\gamma)}{3 + \exp(-2\gamma)}k.
\end{equation}
It was claimed in \cite{lokiang} that player A gains advantage over player B in the quantum game, and the first-mover advantage is being enhanced as the entanglement parameter $\gamma$ goes to infinity. 
\section{Comment on the existing result} 
Studying possible Nash equilibria is a challenging task in a duopoly example with incomplete information. Applying differential calculus in the case of Stackelberg's duopoly requires imposing conditions under which determined Nash equilibria are valid. The assumption (\ref{classicassumption}) already concerning the Nash equilibrium outcome in the classical game (\ref{payoffA})-(\ref{payoffB}) turns out to be highly restrictive. It increases player A's incomplete information about the marginal costs as it puts some restrictions on the probability distribution $(\theta, 1 - \theta)$. In particular, when the distance between $k_{H}$ and $k_{L}$ is sufficiently large, we see from (\ref{classicassumption}) that the probability $\theta$ of player B's marginal cost $c_{BH}$ is close to zero, and the game becomes approximately a game with complete information. 

Another issue concerns restrictions imposed on the parameters $k_{L}$, $k_{H}$ and $\theta$ in the quantum game. The constraint (\ref{quantumrestriction}) was supposed to guarantee non-negative values of  (\ref{qsolution2}) and consequently (\ref{qsolution3}). But (\ref{quantumrestriction}) should read 
\begin{equation}\label{correctcondition}
k_{L} - \frac{\cosh\gamma \exp \gamma}{1+ \cosh\gamma \exp(-\gamma)}k \geq 0 \Leftrightarrow k_{L}\geq \frac{1+\exp 2\gamma}{3+ \exp(-2\gamma)}k,
\end{equation}
which considerably affects the existing study. Since $k > k_{L}$, inequality (\ref{correctcondition}) makes sense if $(1+ \exp2\gamma)/(3+ \exp(-2\gamma)) < 1$. This implies that the entanglement parameter $\gamma$ has to satisfy $\gamma \leq (1/2)\ln (1 + \sqrt{2}) \approx 0.440687$.  The equilibrium outcome (\ref{qsolution1})-(\ref{qsolution3}) is not well defined for $ \gamma >  (1/2)\ln (1 + \sqrt{2})$. This is in agreement with the previous work \cite{previousstackelberg}, where it was shown that the equilibrium outcome obtained by applying differential calculus is restricted by the same condition for $\gamma$. As a result, all the findings of \cite{lokiang} for $ \gamma >  (1/2)\ln (1 + \sqrt{2})$ are not valid. It is still an open question, how the quantum entanglement affects the players' strategic positions in the quantum Stackelberg duopoly with incomplete information for larger values of $\gamma$, especially in the most interesting case when $\gamma \to \infty$.

\section{Classical Stackelberg duopoly with incomplete information}
Stackelberg's duopoly is derived from Cournot's duopoly. Both games have the same range of players' actions and the payoff functions. The only difference is that  two players move simultaneously in the Cournot model whereas they choose sequentially one after the other in the Stackelberg model. Regarding the method of finding equilibrium strategies, differential calculus is a sufficient tool to determine players' optimal choices in both types of duopoly problems. It is also useful to find Bayesian equilibrium in the Cournot duopoly with incomplete information (see, for example, \cite{barron}). However, as it was shown in \cite{lokiang}, that method fails to determine Bayesian equilibrium in the Stackelberg model with incomplete information for the full range of $\theta$. 

In what follows, we work out a best response analysis to find Bayesian equilibria in the classical game. 
\subsection{Formal definition of the game}
Let us first recall the formal description of Stackelberg's duopoly with incomplete information. It is assumed that the marginal cost of player B is either $c_{BH}$, or $c_{BL}$, and $c_{BH} \ne c_{BL}$. The cost $c_{A}$ of player A is commonly known. The cost for player B is considered as a random variable to player~A. Player A only knows that it is $c_{BH}$ with probability $\theta$ or $c_{BL}$ with probability $1-\theta$. Player B knows her own marginal cost. The payoff functions of the players are
\begin{align}\label{ppayoffA}
u_{A}(x_{A}, x_{BH}, x_{BL}) &= \theta u_{A}(x_{A}, x_{BH}) + (1-\theta) u_{A}(x_{A}, x_{BL}),\\ \label{ppayoffB}
u_{B}(x_{A}, x_{BH}, x_{BL}) &= \theta u_{BH}(x_{A}, x_{BH}) + (1-\theta)u_{BL}(x_{A}, x_{BL}),
\end{align}
where
\begin{equation}\label{realpayoffA}
u_{A}(x_{A}, x_{B\cdot}) = \begin{cases} 
x_{A}(a-c_{A} - x_{A} - x_{B\cdot}) &\text{if}~x_{A} + x_{B\cdot} < a, \\ 
-c_{A}x_{A} &\text{if}~x_{A} + x_{B\cdot} \geq a,
\end{cases}
\end{equation}
\begin{equation}\label{realpayoffB}
u_{B\cdot}(x_{A}, x_{B\cdot}) = \begin{cases}
x_{B\cdot}(a-c_{B\cdot} - x_{A} - x_{B\cdot}) &\text{if}~x_{A} + x_{B\cdot} < a, \\
-c_{B\cdot}x_{B\cdot} &\text{if}~x_{A} + x_{B\cdot} \geq a,
\end{cases}
\end{equation}
and $B\cdot$ is either $BH$ or $BL$.
The payoff functions (\ref{realpayoffA}) and  (\ref{realpayoffB}) are commonly used forms to define problems of duopoly (see, for example, \cite{barron} \cite{peters}). The expression $a - x_{A} - x_{B\cdot}$ is defined as the price of the product. The higher the total quantity $x_{A} + x_{B\cdot}$ is, the lower the price of the product. Piecewise-defined functions (\ref{realpayoffA}) and  (\ref{realpayoffB})  take into account the fact that the price cannot be a negative value.
\subsection{Subgame perfect equilibrium} \label{sectionsubgame}
Before we proceed  to find players' optimal strategies, we need to select a proper solution concept. In what follows, we justify applying subgame perfect equilibrium. Initially, the definition of the game suggests using the concept of Bayesian equilibrium. It follows from the fact that the Stackelberg duopoly with (\ref{ppayoffA}) and  (\ref{ppayoffB})  can be written in terms of games with incomplete information. In this terminology, player B has two types, $c_{BH}$ and $c_{BL}$. Bayesian equilibrium is, in fact, equivalent to a standard Nash equilibrium provided that each type of a player occurs with a positive probability \cite{harsai}. Since the Stackelberg duopoly has an extensive structure, Nash equilibrium may result in many strategy profiles, and some of them may exhibit non-optimal behavior off the equilibrium path.

A basic Nash equilibrium refinement that is aimed at excluding non-credible strategies is a subgame perfect equilibrium. It becomes particularly significant when a game has a large number of subgames as it imposes additional conditions beyond the conditions defining the Nash equilibrium. In contrast, when the only subgame is a game itself, a subgame perfect equilibrium is equivalent to a Nash equilibrium, and perfect Bayesian equilibrium may turn out to be a better choice. Perfect Bayesian equilibrium is a counterpart of subgame perfect equilibrium in games with incomplete information. Therefore, it is a suitable Nash equilibrium refinement in the duopoly problem given by (\ref{ppayoffA}) and  (\ref{ppayoffB}). Interestingly, the Stackelberg duopoly problem with incomplete information can be easily converted into a game with perfect information (with a chance move) as it was shown in Fig.~1.  Since player A has no information about an action chosen by the chance mover, it makes no difference to player A whether she chooses her action $x_{A}$ before or after the chance move $c_{L}$ or $c_{H}$. It follows that perfect Bayesian equilibrium does not refine perfect subgame equilibrium in the 
\begin{figure}[t]\label{fig1}
\centering\includegraphics[width=6.5in]{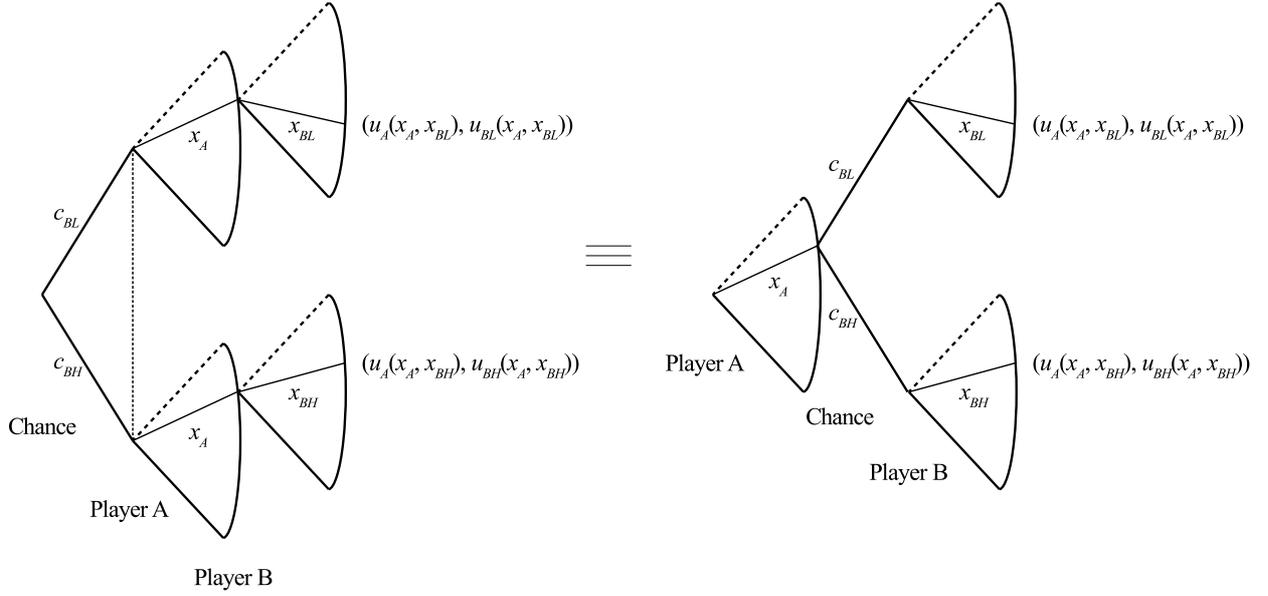}
\caption{Equivalent extensive forms of the Stackelberg duopoly with incomplete information. The extensive form on the right describes a game with perfect information as all information sets in the game consist of a single vertex. }
\end{figure}
Stackelberg duopoly problem with incomplete information.

Similarly to \cite{lokiang}, let $k_{H} = a-c_{H}$, $k_{L} = a-c_{L}$ and $k = \theta k_{H} + (1-\theta)k_{L}$ for $\theta \in (0,1)$ in (\ref{realpayoffA}) and  (\ref{realpayoffB}), and assume that $k_{H} > k_{L} > 0$. To find a subgame perfect equilibrium, we begin with subgames starting from player B's decision nodes. Consider the subgame determined by the chance action $c_{BL}$. Given $x_{A} \leq k_{L}$, player B's best reply is the solution of the equation $\partial u_{BL}(x_{A}, x_{BL})/\partial x_{BL} = 0$ for $x_{A} + x_{BL} < a$, which is $x^*_{BL} = (k_{L} - x_{A})/2$. If $x_{A} > k_{L}$, the derivative of $u_{BL}(x_{A}, x_{BL})$ with respect to $x_{BL}$ is negative. Hence $x^*_{BL} = 0$ in that case. We apply the same argument again with $x_{BL}$ replaced by $x_{BH}$, and obtain
\begin{equation}\label{bestreplies}
x^*_{BL} = \begin{cases} \frac{1}{2}(k_{L} - x_{A})  &\text{if}~ x_{A} \leq k_{L},\\
0 &\text{if}~x_{A} > k_{L}, \end{cases} \quad x^*_{BH} = \begin{cases} \frac{1}{2}(k_{H} - x_{A})  &\text{if}~ x_{A} \leq k_{H},\\
0 &\text{if}~x_{A} > k_{H}. \end{cases} 
\end{equation}
Given best reply functions (\ref{bestreplies}) of player B, the optimal choice of player A is obtained by maximizing the function $x_{A} \mapsto u_{A}(x_{A}, x^*_{BH}, x^*_{BL})$. By taking into account (\ref{bestreplies}),  the expression $u_{A}(x_{A}, x^*_{BH}, x^*_{BL})$ can be written as
\begin{align}\label{payoffAbest}
u_{A}(x_{A}, x^*_{BH}, x^*_{BL}) &= \begin{cases} u_{A}\left(x_{A}, \frac{1}{2}(k_{H} - x_{A}) ,  \frac{1}{2}(k_{L} - x_{A})\right)  &\text{if}~x_{A} \leq k_{L},\\
u_{A}\left(x_{A}, \frac{1}{2}(k_{H} - x_{A}) ,  0 \right)  &\text{if}~ k_{L} < x_{A} \leq k_{H}, \\ 
u_{A}\left(x_{A}, 0 ,  0 \right) &\text{if}~ k_{H} < x_{A} \leq a,
\end{cases} \nonumber \\
&=\begin{cases} \frac{1}{2}(k-x_{A})x_{A}  &\text{if}~x_{A} \leq k_{L}, \\
\frac{1}{2}x_{A}(2k -\theta k_{H}+(\theta - 2)x_{A}) &\text{if}~k_{L} < x_{A} \leq k_{H}, \\ 
x_{A}(k-x_{A}) &\text{if}~k_{H} < x_{A}\leq a.
\end{cases}
\end{align}
Having found $u_{A}(x_{A}, x^*_{BH}, x^*_{BL})$, we are now in a position to determine player A's best reply to $(x^*_{BH}, x^*_{BL})$. The result is presented by the following proposition:
\begin{proposition}
Player A's best reply to $(x^*_{BH}, x^*_{BL})$ given by (\ref{bestreplies}) is 
\begin{equation}\label{xabestreply}
x^*_{A} = \begin{cases} 
\frac{2k_{L}(1-\theta) + \theta k_{H}}{2(2-\theta)} &\text{if}~k_{L} < \frac{\theta k_{H}}{2}\\
k_{L} &\text{if}~\frac{\theta k_{H}}{2} \leq k_{L} \leq \frac{\theta k_{H}}{1+\theta} \\ 
\frac{1}{2}k &\text{if}~ k_{L} > \frac{\theta k_{H}}{1+\theta}, 
\end{cases}
\end{equation}
\end{proposition}
\begin{proof}
We first note that $x_{A} \geq k_{H}$ is not optimal. Player A can obtain a positive payoff by playing $x_{A} \leq k_{L}$ whereas $x_{A} \geq k_{H}$ always results in a negative payoff.
Let us find global maximum points of the first two sub-functions of (\ref{payoffAbest}). Write
\begin{equation}
g(x) = \frac{1}{2}(k-x)x, \quad h(x) = \frac{1}{2}x(2k - \theta k_{H} + (\theta - 2)x).
\end{equation}
An easy computation shows that 
\begin{equation}\label{argmax}
\argmax_{x\in [0, k_{L}]} g(x) = \begin{cases} \frac{1}{2}k &\text{if}~k_{L} > \frac{\theta k_{H}}{1+\theta} \\ 
k_{L} &\text{if}~k_{L} \leq \frac{\theta k_{H}}{1+\theta}
\end{cases}, \quad  \argmax_{x\in [k_{L}, k_{H}]} h(x) = \begin{cases}  
\frac{2k_{L}(1-\theta) + \theta k_{H}}{2(2-\theta)} &\text{if}~k_{L} < \frac{\theta k_{H}}{2} \\
k_{L} &\text{if}~k_{L} \geq \frac{\theta k_{H}}{2}.
\end{cases}
\end{equation}
and $g(k_{L}) = h(k_{L})$. Consider the case $k_{L} > \theta k_{H} / (1+\theta)$. Then $k_{L} > \theta k_{H} /2$, and it follows from (\ref{argmax}) that
\begin{equation}
\max_{x\in [0, k_{L}]}g(x) = g\left(\frac{1}{2}k\right) > g(k_{L}) = h(k_{L}) = \max_{x\in [k_{L}, k_{H}]} h(x).
\end{equation}
We thus proved that $x^*_{A} = k/2$ maximizes (\ref{payoffAbest}) if $k_{L} > \theta k_{H} / (1+\theta)$.

Let us now examine the interval $\theta k_{H}/2 \leq k_{L} \leq \theta k_{H}/(1+\theta)$. Again, by (\ref{argmax}) we conclude that
\begin{equation}
\max_{x \in [0, k_{L}]}g(x) = g(k_{L}) = h(k_{L}) = \max_{x\in [k_{L}, k_{H}]} h(x). 
\end{equation}
Therefore $x^*_{A} = k_{L}$ in that case. In the same manner we can see that $x^*_{A} =  (2k_{L}(1-\theta) + \theta k_{H})/(2(2-\theta))$ if $k_{L} < \theta k_{H}/2$.
\end{proof}
Given best reply (\ref{xabestreply}) to $x^*_{BH}$ and $x^*_{BL}$ we derive subgame perfect equilibrium outcome. 

If $k_{L} > \theta k_{H} / (1+\theta)$, which is equivalent to $k/2 < k_{L}$, then by (\ref{bestreplies}) and (\ref{xabestreply})
\begin{equation}
x^*_{A} = \frac{1}{2}k, \quad x^*_{BH} = \frac{1}{2}\left(k_{H}- \frac{1}{2}k\right), \quad x^*_{BL} = \frac{1}{2}\left(k_{L}- \frac{1}{2}k\right). 
\end{equation}
Similarly, if $\theta k_{H}/2 \leq k_{L} \leq \theta k_{H}/(1+\theta)$, then
\begin{equation}
x^*_{A} = k_{L}, \quad x^*_{BH} = \frac{1}{2}\left(k_{H}- k_{L}\right), \quad x^*_{BL} =0.
\end{equation}
For $k_{L} < \theta k_{H}/2$ we have
\begin{equation}
x^*_{A} = \frac{2(1-\theta)k_{L} 
+ \theta k_{H}}{2(2-\theta)} < \frac{\theta(1-\theta) k_{H} + \theta k_{H}}{2(2-\theta)}  = \frac{\theta k_{H}}{2} < k_{H}.
\end{equation}
Likewise, we can see that $x^*_{A} > k_{L}$. Therefore, 
\begin{equation}
x^*_{BH} = \frac{1}{2}\left(k_{H} -  \frac{2(1-\theta)k_{L} 
+ \theta k_{H}}{2(2-\theta)}\right), \quad x^*_{BL} = 0.
\end{equation}
To summarize,
\begin{equation}\label{equilibriumoutcome}
(x^*_{A}, x^*_{BH}, x^*_{BL}) = \begin{cases}
\left(\frac{1}{2}k,  \frac{1}{2}\left(k_{H} - \frac{1}{2}k\right), \frac{1}{2}\left(k_{L} - \frac{1}{2}k\right)\right) &\text{if}~k_{L} > \frac{\theta k_{H}}{1+\theta}, \\
\left(k_{L}, \frac{1}{2}(k_{H} - k_{L}), 0\right) &\text{if}~\frac{\theta k_{H}}{2} \leq k_{L} \leq \frac{\theta k_{H}}{1+\theta}, \\ 
\left(\frac{2(1-\theta)k_{L} 
+ \theta k_{H}}{2(2-\theta)}, \frac{1}{2}\left(k_{H} -  \frac{2(1-\theta)k_{L} 
+ \theta k_{H}}{2(2-\theta)}\right), 0\right) &\text{if}~k_{L} < \frac{\theta k_{H}}{2}.
\end{cases}
\end{equation}
Formula (\ref{equilibriumoutcome}) leads to the following equilibrium payoff outcomes
\begin{equation}
u_{A}(x^*_{A}, x^*_{BH}, x^*_{BL}) = \begin{cases} 
\frac{1}{8}k^2 &\text{if}~k_{L} > \frac{\theta k_{H}}{1+\theta},\\
\frac{1}{2}\theta(k_{H}-k_{L})k_{L} &\text{if}~\frac{\theta k_{H}}{2} \leq k_{L} \leq \frac{\theta k_{H}}{1+\theta},\\
\frac{(-2k_{L}(-1+\theta) + \theta k_{H})^2}{8(2-\theta)} &\text{if}~k_{L} < \frac{\theta k_{H}}{2}.
\end{cases}
\end{equation}
\begin{equation}
u_{BH}(x^*_{A}, x^*_{BH}) = \begin{cases} 
\frac{1}{16}(k - 2k_{H})^2 &\text{if}~ k_{L} > \frac{\theta k_{H}}{1+\theta}, \\
\frac{1}{4}(k_{H} - k_{L})^2 &\text{if}~\frac{\theta k_{H}}{2} \leq k_{L} \leq \frac{\theta k_{H}}{1+\theta}, \\
\frac{(k+k_{L} + 2k_{H}(-2+\theta) - \theta k_{L})^2}{16(-2+\theta)^2} &\text{if}~k_{L} < \frac{\theta k_{H}}{2}.
\end{cases}
\end{equation}
\begin{equation}
u_{BL}(x^*_{A}, x^*_{BL}) = \begin{cases}
\frac{1}{16}(k-2k_{L})^2 &\text{if}~k_{L} > \frac{\theta k_{H}}{1+\theta} \\
0 &\text{if}~k_{L} \leq \frac{\theta k_{H}}{1+\theta}. 
\end{cases}
\end{equation}
Hence,
\begin{align}\label{ceo}
u_{B}(x^*_{A}, x^*_{BH}, x^*_{BL}) &= \theta u_{BH}(x^*_{A}, x^*_{BH}) + (1-\theta)u_{BL}(x^*_{A}, x^*_{BL}) \nonumber\\ 
&= \begin{cases}  \frac{1}{16}\left((k-2k_{L})^2(1-\theta) + (k-2k_{H})^2\theta\right) &\text{if}~k_{L} > \frac{\theta k_{H}}{1+\theta},\\
\frac{1}{4}(k_{H} - k_{L})^2\theta &\text{if}~\frac{\theta k_{H}}{2} \leq k_{L} \leq \frac{\theta k_{H}}{1+\theta},\\
\frac{(k+k_{L} + 2k_{H}(-2+\theta) - \theta k_{L})^2\theta}{16(-2+\theta)^2} &\text{if}~k_{L} < \frac{\theta k_{H}}{2}.
\end{cases}
\end{align}
We thus obtained a general solution based on subgame perfect equilibrium concept, which is not restricted by any specific relations between $k_{L}$ and $k_{H}$. We see at once the subfunctions of (\ref{equilibriumoutcome})-(\ref{ceo}) defined for $k_{L} > k_{H}/(1+\theta)$ coincide with (\ref{classicsolutionfake}) and (\ref{classicsolutionfakeo}). Moreover, in contrast to what was stated in \cite{lokiang}, a solution is also found for $k_{L} < k/2$, or equivalently for $k_{L} < \theta k_{H}/(1+\theta)$.

It was shown in \cite{lokiang} that whether player A or player B has a better strategic position depends on values of $k_{H}, k_{L}$ and $\theta$. This is also true in the general case. For example, let $k_{H}, k_{L}$ and $\theta$ such that
\begin{equation}\label{roboczyprzedzial}
\frac{\theta k_{H}}{2} < k_{L} < \frac{\theta k_{H}}{1+\theta}.
\end{equation}
Then, the inequality $u_{A}(x^*_{A}, x^*_{BH}, x^*_{BL}) - u_{B}(x^*_{A}, x^*_{BH}, x^*_{BL})$ is equivalent to $k_{L} > k_{H}/3$. It follows that $\theta = 2/3$ implies the first-mover advantage for (\ref{roboczyprzedzial}), whereas for $\theta = 1/2$, it is the second player that obtains a higher equilibrium payoff. 

As we will see in Section~\ref{section6} the maximally correlated quantities of the players in the quantum game lead to a much more transparent equilibrium outcome. 
\section{Li-Du-Massar approach to Stackelberg's duopoly with incomplete information}
\subsection{Li-Du-Massar approach to duopoly problems}
For the convenience of the reader we repeat the relevant material from \cite{lidumassar},  \cite{fracorbertrand}, \cite{fracorduopoly1}  and \cite{fracorduopoly2} in order to make our exposition self-contained. 

Let $|00\rangle$ be the initial state and $J(\gamma) = e^{-\gamma(a^{\dag}_{A}a^{\dag}_{B} - a_{A}a_{B})}$ be a unitary operator, where $\gamma\geq 0$ and $a^{\dag}_{i}$ ($a_{i}$) represents the creation (annihilation) operator of electromagnetic field $i$. The player $i$'s strategies are unitary operators of the form 
\begin{equation}\label{listrategies}
D_{i}(x_{i}) = e^{x_{i}(a^{\dag}_{i} - a_{i})/ \sqrt{2}}, x_{i} \in [0,\infty), i=A,B. 
\end{equation}
The operator $J(\gamma)$ and the strategy profile $D_{A}(x_{A})\otimes D_{B}(x_{B})$ determine the final state $|\Psi_{\mathrm{f}}\rangle$,
\begin{equation}
|\Psi_{\mathrm{f}}\rangle = J^{\dag}(\gamma)(D_{A}(x_{A})\otimes D_{B}(x_{B}))J(\gamma)|00\rangle.
\end{equation}
The quantity $q_{i}$ is then obtained by performing the measurement $X_{i} = \left(a^{\dag}_{i} + a_{i}\right)/\sqrt{2}$ on the state $|\Psi_{\mathrm{f}}\rangle$. The result is
\begin{align}\label{licorrelation}\begin{split}
&q_{A} = \langle \Psi_{\mathrm{f}} | X_{A} |\Psi_{\mathrm{f}}\rangle = x_{A}\cosh{\gamma} + x_{B}\sinh{\gamma}, \\ &q_{B} = \langle \Psi_{\mathrm{f}} | X_{B} |\Psi_{\mathrm{f}}\rangle = x_{B}\cosh{\gamma} + x_{A}\sinh{\gamma}.
\end{split}
\end{align}
We obtain the quantum extension of the classical Stackelberg duopoly by substituting (\ref{licorrelation}) into (\ref{realpayoffA}) and (\ref{realpayoffB}), 
\begin{equation}
u_{A(B\cdot)}(x_{A}, x_{B\cdot}, \gamma) = \begin{cases}q_{A(B\cdot)}(a-c_{A(B\cdot)}-e^{\gamma}(x_{A} + x_{B\cdot})) &\mbox{if}~e^{\gamma}(x_{A} + x_{B\cdot}) \leq a,\\
-c_{A(B\cdot)}q_{A(B\cdot)} &\mbox{if}~e^{\gamma}(x_{A} + x_{B\cdot}) > a.
  \end{cases}
  \end{equation}
It is worth pointing out that the resulting outputs (\ref{licorrelation}) are not in units of $x_{i}$'s. Given $x_{A}$ and $x_{B}$ fixed, we see at once that $q_{i}$ increases with $\gamma$, for $i=A,B$. We can normalize (\ref{licorrelation}) by setting
\begin{equation}\label{normalize}
x_{i} \mapsto D\left(\frac{x_{i}}{e^{\gamma}}\right).
\end{equation}
It follows easily from (\ref{normalize}) that the resulting quantities become
\begin{equation}\label{ncorrelation}
q'_{A} = \frac{x_{A}\cosh\gamma + x_{B}\sinh\gamma}{e^{\gamma}}, \quad q'_{B} = \frac{x_{B}\cosh\gamma + x_{A}\sinh\gamma}{e^{\gamma}},
\end{equation}
Both (\ref{licorrelation}) and (\ref{ncorrelation}) are equivalent when studying duopoly examples by the Li-Du-Massar scheme. For example, applying (\ref{ncorrelation}) in the Cournot duopoly \cite{fracorduopoly1} results in the unique Nash equilibrium $(x_{A}, x_{B})$ such that 
\begin{equation}\label{mojcournot}
x^*_{A} = x^*_{B} = \frac{a-c}{3 + \tanh\gamma}.
\end{equation}
In the case of applying  (\ref{licorrelation}) the Nash equilibrium strategy is of the form
\begin{equation}
x^*_{i} = \frac{(a-c)\cosh\gamma}{1+2e^{2\gamma}}, 
\end{equation}
which is simply the division of  (\ref{mojcournot}) by $e^{\gamma}$. It is worth pointing out that using (\ref{ncorrelation}) enables us to compare classical and quantum equilibria without referring to payoff outcomes. Quantity (\ref{mojcournot}) ranges from the classical equilibrium strategy $(a-c)/3$ to the part of the Pareto-optimal profile $(a-c)/4$. Another advantage of applying (\ref{ncorrelation}) is that we see how $x_{A}$ and $x_{B}$ are correlated when $\gamma \to \infty$, 
\begin{equation}\label{bordercase}
\lim_{\gamma \to \infty}q'_{i} = \frac{x_{A} + x_{B}}{2}.
\end{equation}
Equality (\ref{bordercase}) will be particularly useful in the next section.
\subsection{Subgame perfect equilibria in maximally entangled game} \label{section6}
We proceed with the study of the Stackelberg duopoly with incomplete information in the quantum domain. The work \cite{lokiang}  provides us with the equilibrium outcome for $\gamma \leq \gamma_{0} = (1/2)\ln (1 + \sqrt{2})$. Our initial investigation showed that the results for $\gamma$ slightly higher than $\gamma_{0}$ are moderately interesting. The resulting equilibrium payoffs are very complex, and the players' strategic positions depend on the values of $k_{H}, k_{L}$ and $\theta$. For this reason, we focus on the case in which $\gamma \to \infty$.

The common method to find a Nash equilibrium (or its refinement) in the Li-Du-Massar approach to a duopoly problem with maximally correlated quantities of the players is determining the equilibrium profile for a general value of $\gamma$, and then taking the limit as $\gamma$ goes to infinity. In fact, we can simplify the analysis by substituting  (\ref{bordercase}) into (\ref{realpayoffA}) and (\ref{realpayoffB}). In this way we obtain players' payoff functions in the maximally correlated case.
\begin{equation}\label{qgame1}
u_{A}(x_{A}, x_{B\cdot}) = \begin{cases} 
\frac{x_{A} + x_{B\cdot}}{2}(k-x_{A} - x_{B\cdot}) &\text{if}~x_{A} + x_{B\cdot} \leq a \\
\frac{-c(x_{A} + x_{B})}{2} &\text{if}~x_{A} + x_{B\cdot} > a,
\end{cases}
\end{equation}
where $B\cdot$ is either $BH$ or $BL$, and
\begin{align} \label{ubh}
u_{BH}(x_{A}, x_{BH}) &= \begin{cases}  \frac{x_{A} + x_{BH}}{2}(k_{H} - x_{A} - x_{BH}) &\text{if}~x_{A} + x_{BH} \leq a \\ 
\frac{-c_{H}(x_{A} + x_{BH})}{2} &\text{if}~x_{A} + x_{BH}>a,
\end{cases} \\ 
u_{BL}(x_{A}, x_{BL}) &= \begin{cases}  \frac{x_{A} + x_{BL}}{2}(k_{L} - x_{A} - x_{BL}) &\text{if}~x_{A} + x_{BL} \leq a \\ 
\frac{-c_{L}(x_{A} + x_{BL})}{2} &\text{if}~x_{A} + x_{BL}>a,
\end{cases} \label{qgame3}
\end{align}
Now, we are left with the task of determining players' rational choices with respect to (\ref{qgame1})-(\ref{qgame3}). 
The method for finding subgame perfect equilibria is similar to that used in Subsection~\ref{sectionsubgame}. We first find player B's best reply to a strategy $x_{A}$ of player A. For $x_{A} + x_{BH} \leq a$ we have
\begin{equation}
\frac{\partial u_{BH}}{\partial x_{BH}} = \frac{1}{2}k_{H} - x_{A} - x_{BH},  \quad \frac{\partial^2 u_{BH}}{\partial x_{BH}^2} = -1.
\end{equation}
It follows that $x^*_{BH} = k_{H}/2 - x_{A}$ if $x_{A} \leq k_{H}/2$. For $x_{A} > k_{H}/2$, we have $\partial u_{BH}/ \partial x_{BH} < 0$. Therefore, $u_{BH}$ has the highest value at $x_{BH} = 0$. The same reasoning applies to $x^*_{BL}$. We thus get
\begin{equation}\label{brq1}
x^*_{BH} = \begin{cases}  
\frac{1}{2}k_{H} - x_{A} &\text{if}~x_{A} \leq \frac{1}{2} k_{H}, \\
0 &\text{if}~x_{A} > \frac{1}{2} k_{H},
\end{cases}  \quad x^*_{BL} = \begin{cases}  
\frac{1}{2}k_{L} - x_{A} &\text{if}~x_{A} \leq \frac{1}{2} k_{L}, \\
0 &\text{if}~x_{A} > \frac{1}{2} k_{L}.
\end{cases} 
\end{equation}
We are now in a position to determine player A's best reply $x^*_{A}$ to $(x^*_{BH}, x^*_{BL})$.
By substituting (\ref{brq1}) into (\ref{ppayoffA}) we can write $u_{A}(x_{A}, x^*_{BH}, x^*_{BL})$ as
\begin{align}
u_{A}(x_{A}, x^*_{BH}, x^*_{BL}) &= \begin{cases} 
u_{A}\left(x_{A}, \frac{1}{2}k_{H} - x_{A}, \frac{1}{2}k_{L} - x_{A}\right) &\text{if}~x_{A} \leq \frac{1}{2}k_{L}, \\
u_{A}(x_{A},  \frac{1}{2}k_{H} - x_{A}, 0) &\text{if}~\frac{1}{2}k_{L} < x_{A} \leq \frac{1}{2}k_{H}, \\
u_{A}(x_{A}, 0, 0) &\text{if}~\frac{1}{2}k_{H} < x_{A} \leq a.
\end{cases}
\\ &= \begin{cases}
\frac{k^2 - (2k^2-2kk_{H} + k^2_{H})\theta}{8(1 - \theta)} &\text{if}~x_{A} \leq \frac{1}{2}k_{L}, \\
\frac{1}{8}((2k-k_{H})\theta k_{H} + 4k(1-\theta)x_{A} - 4(1-\theta)x^2_{A})  &\text{if}~ \frac{1}{2}k_{L} < x_{A} \leq \frac{1}{2}k_{H},\\ 
\frac{x_{A}}{2}(k-x_{A}) &\text{if}~\frac{1}{2}k_{H} < x_{A} \leq a.
\end{cases}
\end{align}
It enables us to formulate 
\begin{proposition}
A subgame perfect equilibrium outcome in game defined by (\ref{qgame1})-(\ref{qgame3}) is 
\begin{equation}\label{mainprofile}
\left(x^*_{A}, x^*_{BH}, x^*_{BL}\right) = \left(\frac{k}{2}, \frac{1}{2}(k_{H} - k), 0\right).
\end{equation}
\end{proposition}
\begin{proof}
Let 
\begin{equation}
g(x) = \frac{1}{8}((2k-k_{H})\theta k_{H} + 4k(1-\theta)x_{A} - 4(1-\theta)x^2_{A}), \quad h(x) = \frac{x}{2}(k - x).
\end{equation}
Then 
\begin{equation}
\frac{dg}{dx} = 0 \Leftrightarrow x = \frac{k}{2} \quad \text{and} \quad \frac{d^2g}{dx^2} <0.
\end{equation}
It follows that $\argmax_{k_{L}/2 < x \leq k_{H}/2} g(x) = k/2$. Furthermore, it is easy to check that $\argmax_{k_{H}/2 < x < a}h(x) = k_{H}/2$.
Since
\begin{equation}
g\left(\frac{k}{2}\right) = \frac{1}{8}\left(k^2 - (k-k_{H})^2\theta\right) > \frac{1}{8}k_{H}(2k -k_{H}) = h\left(\frac{k_{H}}{2}\right)
\end{equation}
and 
\begin{equation}
g\left(\frac{k}{2}\right) > \frac{k^2 - (2k^2-2kk_{H} + k^2_{H})\theta}{8(1 - \theta)}
\end{equation}
for every $\theta \in (0,1)$, we conclude that 
\begin{equation}
\argmax_{0 \leq x_{A} \leq a}u_{A}(x_{A}, x^*_{BH}, x^*_{BL}) = \frac{k}{2}.
\end{equation}
Now substituting $x^*_{A} = k/2$ into (\ref{brq1}) completes the proof. 
\end{proof}
Let us determine the payoff outcome corresponding to (\ref{mainprofile}). According to (\ref{ppayoffA}) and (\ref{qgame1}), we obtain
\begin{equation}\label{finalpayoffA}
u_{A}\left(x^*_{A}, x^*_{BH}, x^*_{BL}\right) = \frac{1}{8}\left(k^2 - (k-k_{H})^2\theta\right).
\end{equation}
From (\ref{ubh}) and (\ref{qgame3}) we see that
\begin{equation}
u_{BH}(x^*_{A}, x^*_{BH}) = \frac{1}{8}k^2_{H}, \quad u_{BL}(x^*_{A}, x^*_{BL}) = \frac{1}{8}\left(k^2_{L} - (k_{H} - k_{L})^2\theta^2\right).
\end{equation}
Now (\ref{ppayoffB}) becomes
\begin{equation}\label{finalpayoffB}
\theta u_{BH}(x^*_{A}, x^*_{BH}) + (1-\theta)u_{BL}(x^*_{A}, x^*_{BL}) = \frac{1}{8}\left(k^2 + (k-k_{H})^2\theta\right).
\end{equation}
From what has already been obtained, it may be concluded that player B gains advantage in the game with incomplete information. The difference between player B and player A's equilibrium payoffs is positive for every $\theta \in (0,1)$ (see, Fig.~\ref{figure2}), 
\begin{equation}
(u_{B}-u_{A})(x^*_{A}, x^*_{BH}, x^*_{BL}) = \frac{1}{4}\left(k_{H} - k_{L}\right)^2(1- \theta)^2\theta.
\end{equation}
\begin{figure}[t]
\centering\includegraphics[width=4in]{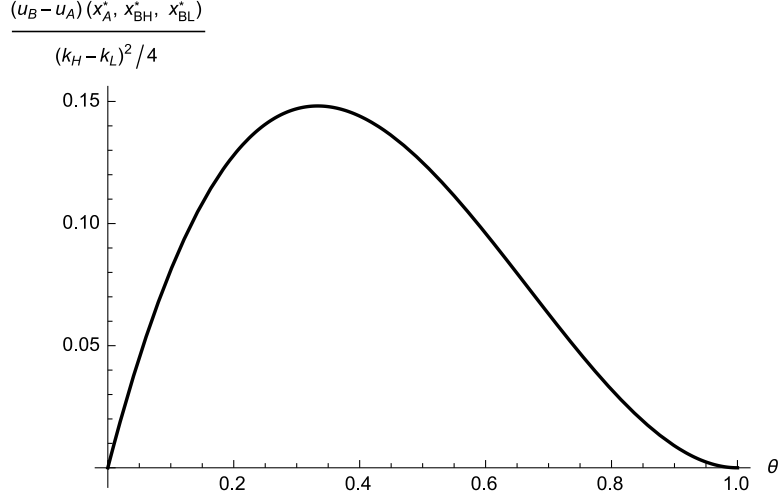}
\caption{The difference of players' equilibrium payoffs for $\theta \in [0,1]$.}
\label{figure2}
\end{figure}
The results (\ref{finalpayoffA}) and (\ref{finalpayoffB}) generalize those of \cite{fracorduopoly2}. If $\theta \in \{0,1\}$, player A has complete information about player B's marginal cost. In that case both (\ref{finalpayoffA}) and (\ref{finalpayoffB}) equal $k^2/8$, which is the equilibrium payoff in the quantum Stackelberg duopoly (with complete information) \cite{fracorduopoly2}. 
\section{Conclusions}
Li-Du-Massar scheme \cite{lidumassar} has made a significant contribution to quantum game theory. It has shed some new light on duopoly problems in the quantum domain. Although, it was originally designed for static duopoly examples, it has also found application in sequential types of duopoly.  The Stackelberg duopoly is a common example of sequential duopoly models. It is characterized by non-symmetric strategic positions of the players. The problem becomes even more complex if one of the players does not have complete information about some aspects of the game. We reexamined the previous study \cite{lokiang} in which the player who moves first has incomplete information about the marginal cost of the second player. The previous result was restricted to specific values of $k_{H}, k_{L}$ and $\theta$ in both the classical and quantum cases. Moreover, the existing study of equilibrium outcomes turned out to be valid for initial values of the entanglement parameter.  

Our work has provided the subgame perfect equilibrium analysis of the classical and the quantum game without any restrictions on the marginal costs. We showed that each player may gain a strategic advantage in the classical Stackelberg duopoly with incomplete information depending on player B's marginal costs and player A's level of certainty of those marginal costs. Interestingly, our study on the equilibrium outcomes in the quantum game shows a definitive second-mover advantage. We proved that it holds provided that the first player assigns a positive probability to both $c_{H}$ and $c_{L}$, and the players' quantities are maximally correlated. 

\section*{Acknowledgments}

This work was supported by the National Science Centre, Poland under the project 2016/23/D/ST1/01557.

\end{document}